\documentclass[12pt]{article}

\usepackage[margin=1in]{geometry}
\usepackage{amsmath,amssymb,amsthm,mathtools,bm}
\usepackage{booktabs}
\usepackage{enumitem}
\usepackage{natbib}
\usepackage[colorlinks=true,linkcolor=blue,citecolor=blue,urlcolor=blue]{hyperref}
\usepackage[nameinlink,noabbrev]{cleveref}
\usepackage{setspace}
\newcommand{\E}{\mathbb E}
\newcommand{\Var}{\operatorname{Var}}
\newcommand{\Cov}{\operatorname{Cov}}
\newcommand{\AVar}{\operatorname{AVar}}
\newcommand{\diag}{\operatorname{diag}}

\newcommand{\R}{\mathbb R}
\newcommand{\p}{\overset{p}{\to}}
\newcommand{\dlim}{\overset{d}{\to}}

\newtheorem{theorem}{Theorem}
\newtheorem{proposition}{Proposition}
\newtheorem{lemma}{Lemma}
\newtheorem{corollary}{Corollary}
\newtheorem{assumption}{Assumption}

\newtheorem{remark}{Remark}

\crefname{assumption}{Assumption}{Assumptions}
\Crefname{assumption}{Assumption}{Assumptions}

\title{When Does Heteroskedasticity Matter?\\
\large A Contrast-Specific Theory of Robust Inference}
\author{
Ulrich Hounyo\thanks{%
Department of Economics, University at Albany -- State University of New
York, Albany, NY 12222, USA. E-mail: khounyo@albany.edu. }\\
Department of Economics\\
University at Albany, SUNY
}

\date{\today}

\begin{document}
\maketitle

\begin{abstract}
Conventional heteroskedasticity diagnostics ask whether the conditional variance of the regression disturbance varies with covariates. This paper asks a different question: when does that variation matter for inference on the estimand of interest? The paper develops a contrast-specific theory characterizing when covariance perturbations are inferentially relevant. We show that, for any linear contrast $a'\beta$ in a linear regression, the difference between the heteroskedasticity-robust variance and the pooled fixed-design variance is governed by the empirical covariance between conditional error variance and a contrast-specific leverage score. Thus, heteroskedasticity may be present in the model yet first-order irrelevant for a particular coefficient or linear combination. Conversely, modest heteroskedasticity may have a large inferential effect if it is concentrated on observations that are highly informative for the contrast of interest. We characterize the effect exactly through a heteroskedasticity relevance ratio and a standard-error inflation factor, relate the result to pairs and residual bootstrap procedures, and extend the decomposition to general covariance structures, where off-diagonal dependence contributes a separate contrast-specific term. The results provide a unified way to understand why robust, clustered, and bootstrap standard errors can differ across coefficients in the same regression.
\end{abstract}

\vfill{}

\thispagestyle{empty}

\pagebreak{}
\section{Introduction}
\label{sec:introduction}

Applied researchers working with heterogeneous data commonly encounter a puzzling pattern: within the same regression, robust standard errors may differ dramatically from classical standard errors for some coefficients while barely moving for others. A researcher estimating a wage equation might find that the coefficient on education has a robust standard error essentially identical to its classical counterpart, while the coefficient on experience has a robust standard error forty percent larger. Yet both coefficients are estimated from the same regression, with the same heteroskedastic errors, and the same robust covariance matrix. Standard theory gives no guidance on which coefficients will be affected and by how much. Existing diagnostics test for heteroskedasticity in the regression as a whole; none identifies which estimands it actually threatens.

This paper resolves the puzzle. We show that whether heteroskedasticity affects inference on a given coefficient is a property of that coefficient's \emph{relationship to the variance structure}, not a property of the regression model as a whole. The same heteroskedasticity pattern can simultaneously matter a great deal for one estimand and be entirely irrelevant for another. The relevant object is an alignment between the error variance profile and what we call the \emph{directional leverage} of the estimand. Put differently, inference is determined not simply by where heteroskedasticity occurs, but by how the heteroskedasticity pattern aligns with the observations that carry information for the estimand. This alignment perspective is the organizing principle of the paper. In this sense, directional leverage measures where the information for the estimand comes from, whereas heteroskedasticity describes where uncertainty is concentrated. Inference is therefore determined not by where heteroskedasticity occurs, but by where heteroskedasticity occurs relative to the information supporting the estimand. If observations with large conditional error variances receive little informational weight for the estimand, heteroskedasticity may have essentially no first-order effect on inference despite being present throughout the model. Conversely, even modest heteroskedasticity can substantially affect inference when uncertainty is concentrated on the observations that are most informative for the estimand.

\paragraph{The main result.}
For any linear contrast $a'\beta$ in a linear regression with heteroskedastic errors, the difference between the heteroskedasticity-robust variance and the pooled fixed-design variance satisfies the exact identity
\[
    V_H(a)-V_R(a)
    =
    \sum_{i=1}^{n}\ell_i(a)\left(\sigma_i^2-\bar\sigma^2\right)
    =
    n\,\Cov_n\!\left(\ell_i(a),\,\sigma_i^2\right),
\]
where $\ell_i(a)=\{a'(X'X)^{-1}x_i\}^2$ is the \emph{directional leverage} of observation $i$ for the contrast $a'\beta$. Heteroskedasticity inflates inference on $a'\beta$ when observations with above-average error variance are also above-average in directional leverage, and deflates it when the opposite holds. When the two are empirically uncorrelated, heteroskedasticity is present in the data but first-order irrelevant for this particular contrast. This characterization is exact rather than asymptotic: every first-order effect of heteroskedasticity on the contrast operates through this single covariance.

At first sight, one might expect heteroskedasticity to affect every contrast once it is present in the model. Our analysis shows that this intuition is incorrect. The same variance structure can substantially inflate the uncertainty of one contrast while leaving another essentially unchanged, because different contrasts extract information from different subsets of the data.

The central object of the paper is the heteroskedasticity relevance ratio
\[
    R(a)=V_H(a)/V_R(a),
\]
defined as the ratio between the heteroskedastic and pooled variance benchmarks.
Together with the corresponding standard-error inflation factor
\(S(a)=\sqrt{R(a)}\), it quantifies the practical relevance of heteroskedasticity
for the contrast \(a'\beta\). We show that \(R(a)\) has a weighted-average
representation: it equals the contrast-specific information-weighted average of
\(\sigma_i^2/\bar\sigma^2\), where the weights
\(w_i(a)=\ell_i(a)/\sum_j\ell_j(a)\) assign greater weight to observations that
contribute more information to the contrast.  Consequently, the sign and magnitude of \(R(a)-1\) admit an immediate
interpretation: they quantify the extent to which the heteroskedasticity pattern
is \emph{aligned with} or \emph{against} the observations supporting the estimand.

\paragraph{Bootstrap implications.}
The directional identity gives a new interpretation of the residual and pairs bootstrap. The residual bootstrap pools all residuals into a common distribution and then resamples from it, mechanically breaking any association between $x_i$ and $u_i$. The residual bootstrap reproduces the pooled-variance benchmark exactly in the
bootstrap world, yielding the exact conditional finite-sample identity
\[
Var^*(a'\hat\beta_R^*|X)=\widehat V_R(a).
\] The pairs bootstrap resamples $(x_i,y_i)$ jointly, preserving that association asymptotically, and consistently estimates asymptotically $V_H(a)$. Their ratio equals $\widehat R(a)$. Contrary to standard intuition, neither bootstrap is uniformly more conservative: which yields larger standard errors depends on the sign of $\mathrm{Cov}_n(\ell_i(a),\hat u_i^2)$ for the specific contrast under study. We show both in theory and via simulation that there exist natural designs where the residual bootstrap is the more conservative procedure for some contrasts and the pairs bootstrap is for others — within the same data set.

\paragraph{Why this matters even though heteroskedasticity-robust estimators are always valid.}
Since the White estimator is asymptotically valid regardless of the underlying variance structure, one might naturally ask whether the heteroskedasticity relevance ratio $R(a)$ has practical significance beyond simply describing the discrepancy between robust and classical inference. The answer is affirmative. While the White estimator remains asymptotically valid in all cases, $R(a)$ determines when heteroskedasticity has a substantively important effect on inference for the estimand of interest, quantifies the magnitude of that effect, and explains when robust, classical, and bootstrap procedures yield materially different conclusions.  We show this objection, while correct about asymptotic validity, misses two practical consequences. First, validity is not the same as conservatism: applied researchers commonly treat robust standard errors as a conservative alternative to classical ones, but this is only true when $R(a)>1$; when $R(a)<1$, robust inference is \emph{less} conservative than classical inference, and we show in simulation that the classical test can have severely distorted size in this regime. Second, $R(a)$ governs the relative efficiency of robust versus classical standard-error estimates: there is no efficiency cost to using the robust estimator when $R(a)\approx1$, but a real one when $R(a)$ is far from one, connecting our results directly to the efficiency critique of \citet{KingRoberts2015}. \Cref{rem:validity} develops this point with explicit size and efficiency calculations.

\paragraph{Contributions and scope.}
This paper makes four contributions.
Conceptually, we introduce directional leverage and show that it fully characterizes the estimand-specific relevance of heteroskedasticity (\Cref{sec:linear_regression}).
Theoretically, we define the heteroskedasticity relevance ratio $R(a)$ and establish its information-weighting representation and calibration bounds (\Cref{sec:ratio}).
Methodologically, we derive a scalar Lindeberg--Feller central limit theorem for the estimator $\widehat R(a)$, yielding confidence intervals and tests for contrast-specific relevance (\Cref{sec:asymptotics}).
More broadly, we extend the framework to general covariance structures — clustering, serial correlation, spatial and network dependence — where the same alignment principle governs the estimand-specific relevance of dependence (\Cref{sec:dependence}).

\paragraph{Related literature.}
The first strand concerns heteroskedasticity-consistent covariance estimation.
\citet{White1980} showed that valid inference remains possible under unknown
heteroskedasticity; subsequent work has focused on finite-sample improvements,
yielding the HC1--HC5 family and related refinements
\citep{MacKinnonWhite1985,ChesherAustin1991,LongErvin2000,CribariNeto2004};
see \citet{MacKinnon2013} for a survey. Existing work in this literature
primarily asks how to estimate the heteroskedasticity-robust covariance matrix
accurately under increasingly general forms of heteroskedasticity and
dependence. Our question is different. Rather than improving covariance
estimation, we ask when the departure from the classical covariance benchmark
is substantively relevant for a particular inferential objective.

The second strand concerns bootstrap inference. Beginning with
\citet{Freedman1981}, a large literature has studied the pairs, residual, and
wild bootstraps
\citep{Wu1986,Liu1988,Mammen1993,Hall1992,DavidsonMacKinnon2004,DavidsonMacKinnon2006,MacKinnon2013}.
This literature typically takes the asymptotic variance of the limiting
distribution as the inferential target and studies whether bootstrap procedures
consistently estimate that variance, together with their higher-order
properties. Recent work by \citet{HahnLiao2021} further shows that bootstrap
second-moment estimators generally lead to conservative inference under this
conventional framework. Our analysis addresses a complementary and logically
prior question. Rather than taking the variance target as given, we
characterize when different asymptotically valid variance measures are
relevant for a particular estimand. Once the appropriate inferential target has
been identified, the relative behavior of classical and bootstrap standard
errors follows naturally from the heteroskedasticity relevance ratio developed
in this paper.

The third strand concerns robust inference under general dependence,
including HAC, cluster-robust, and spatially robust covariance estimation
\citep{NeweyWest1987,Andrews1991,CameronGelbachMiller2011}. We complement this
literature by showing that the relevance of these covariance structures is
likewise contrast-specific: the same dependence structure can inflate the
variance of one contrast while reducing that of another relative to the
independence benchmark. Thus, the practical importance of dependence, like that
of heteroskedasticity, depends on how the dependence pattern aligns with the
information supporting the estimand.

Finally, \citet{KingRoberts2015} argue that robust standard errors often reveal
model misspecification rather than merely correcting it. Our analysis is
complementary. Conditional on the maintained specification, we characterize
when robust and classical covariance measures differ in ways that are
substantively important for the inferential objective. Consequently, our
framework separates the validity of robust covariance estimation from its
practical relevance for the particular estimand under consideration.

The rest of the paper is organized as follows. \Cref{sec:general_framework} presents a general covariance-relevance principle. \Cref{sec:linear_regression} derives the directional leverage identity. \Cref{sec:ratio} studies the heteroskedasticity relevance ratio and its information-weighting representation. \Cref{sec:bootstrap} interprets the identity through the lens of pairs and residual bootstrap. \Cref{sec:dependence} extends to general covariance structures. \Cref{sec:asymptotics} provides the asymptotic theory for $\widehat R(a)$. \Cref{sec:examples} presents the Monte Carlo study. \Cref{sec:empirical} applies the framework to the Boston Housing data. \Cref{sec:practical} offers practical guidance for applied researchers. \Cref{sec:conclusion} concludes.

\section{A Covariance-Relevance Principle}
\label{sec:general_framework}

We begin with a general observation. Let $\hat\theta$ be an estimator of a finite-dimensional parameter $\theta_0\in\R^p$ satisfying the asymptotic linear representation
\begin{equation}
    \sqrt n(\hat\theta-\theta_0)
    =
    \frac{1}{\sqrt n}\sum_{i=1}^n \psi_i + o_p(1),
\label{eq:AL}
\end{equation}
where $\E(\psi_i)=0$. Let
\[
    \Omega=\AVar\{\sqrt n(\hat\theta-\theta_0)\}.
\]
For a scalar contrast $a'\theta_0$, the asymptotic variance is $a'\Omega a$. If the covariance matrix is written as a baseline component plus a perturbation,
\[
    \Omega=\Omega_0+\Delta,
\]
then the relevance of $\Delta$ for inference on $a'\theta_0$ is governed by the scalar quadratic form $a'\Delta a$.

\begin{proposition}[Covariance-relevance principle]
\label{prop:cov_relevance}
Suppose \eqref{eq:AL} holds and $\Omega=\Omega_0+\Delta$, where $\Omega_0$ and $\Delta$ are symmetric matrices and $\Omega$ is positive semidefinite. Then, for every fixed $a\in\R^p$,
\[
    \AVar\{\sqrt n(a'\hat\theta-a'\theta_0)\}
    =
    a'\Omega_0 a+a'\Delta a.
\]
Consequently, the covariance perturbation $\Delta$ is first-order irrelevant for inference on $a'\theta_0$ if and only if $a'\Delta a=0$.
\end{proposition}

The proof is an immediate consequence of premultiplying the representation \eqref{eq:AL} by $a'$; see \Cref{app:thm1}.

\Cref{prop:cov_relevance} is intentionally simple. Its role is to separate two questions that are often conflated. A covariance perturbation may be large as a matrix, but irrelevant for a particular contrast. Conversely, a perturbation that is small in a global norm may have a large effect on a contrast if it is concentrated in the relevant direction. The remainder of the paper applies this principle to robust inference in linear regression, where the perturbation is induced by heteroskedasticity and dependence.

\section{Linear Regression and Directional Leverage}
\label{sec:linear_regression}

Consider the linear regression model
\begin{equation}
    y_i=x_i'\beta+u_i,
    \qquad i=1,\ldots,n,
\label{eq:model}
\end{equation}
where $x_i$ is a $k\times1$ vector of regressors. Let $X=(x_1,\ldots,x_n)'$ and assume $X'X$ is nonsingular. Throughout this section we work conditionally on $X$ and assume
\[
    \E(u_i\mid X)=0,
    \qquad
    \E(u_i u_j\mid X)=0\quad (i\neq j),
\]
while allowing
\[
    \Var(u_i\mid X)=\sigma_i^2.
\]
Let
\[
    \Sigma=\diag(\sigma_1^2,\ldots,\sigma_n^2),
    \qquad
    \bar\sigma^2=\frac1n\sum_{i=1}^n\sigma_i^2.
\]
The OLS estimator is $\hat\beta=(X'X)^{-1}X'y$.

For a fixed nonzero contrast $a\in\R^k$, define the directional leverage
\begin{equation}
    \ell_i(a)=\{a'(X'X)^{-1}x_i\}^2,
    \qquad i=1,\ldots,n.
\label{eq:ell}
\end{equation}
This is not the usual hat-matrix leverage $h_i=x_i'(X'X)^{-1}x_i$. Classical leverage has long been recognized as important for the finite-sample behavior of heteroskedasticity-consistent covariance estimators
\citep{ChesherAustin1991}. Directional leverage is conceptually different:
it measures the contribution of each observation to the sampling variability
of a particular contrast rather than its influence on fitted values. 

\begin{lemma}[Sum of directional leverages]
\label{lem:sum_ell}
For every $a\in\R^k$,
\[
    \sum_{i=1}^n \ell_i(a)=a'(X'X)^{-1}a.
\]
\end{lemma}

The proof follows immediately from $\sum_i x_ix_i'=X'X$; see \Cref{app:lem1}.

The heteroskedastic fixed-design variance of $a'\hat\beta$ is
\begin{equation}
    V_H(a)=a'(X'X)^{-1}X'\Sigma X(X'X)^{-1}a.
\label{eq:VH}
\end{equation}
The pooled fixed-design variance is
\begin{equation}
    V_R(a)=\bar\sigma^2 a'(X'X)^{-1}a.
\label{eq:VR}
\end{equation}
The notation $R$ is mnemonic for residual-based or pooled-residual inference, because a residual bootstrap based on a single pooled residual distribution targets this variance.

\begin{theorem}[Directional heteroskedasticity identity]
\label{thm:directional_identity}
For every nonzero contrast $a\in\R^k$,
\begin{equation}
    V_H(a)-V_R(a)
    =
    \sum_{i=1}^n \ell_i(a)(\sigma_i^2-\bar\sigma^2).
\label{eq:directional_identity}
\end{equation}
Equivalently,
\begin{equation}
    V_H(a)-V_R(a)
    =
    n\,\Cov_n\{\ell_i(a),\sigma_i^2\},
\label{eq:cov_identity}
\end{equation}
where $\Cov_n$ denotes empirical covariance over $i=1,\ldots,n$.
\end{theorem}

\begin{proof}
Since $\Sigma$ is diagonal,
\[
    V_H(a)=\sum_{i=1}^n\{a'(X'X)^{-1}x_i\}^2\sigma_i^2
    =\sum_{i=1}^n\ell_i(a)\sigma_i^2.
\]
By \Cref{lem:sum_ell},
\[
    V_R(a)=\bar\sigma^2\sum_{i=1}^n\ell_i(a).
\]
Subtracting gives \eqref{eq:directional_identity}. Let $\bar\ell(a)=n^{-1}\sum_i\ell_i(a)$. Because $\sum_i(\sigma_i^2-\bar\sigma^2)=0$,
\[
    \sum_{i=1}^n\ell_i(a)(\sigma_i^2-\bar\sigma^2)
    =
    \sum_{i=1}^n\{\ell_i(a)-\bar\ell(a)\}(\sigma_i^2-\bar\sigma^2)
    =n\Cov_n\{\ell_i(a),\sigma_i^2\}.
\]
\end{proof}

\Cref{thm:directional_identity} is the central result of the paper. It implies that the relevance of heteroskedasticity is contrast-specific. The same heteroskedasticity pattern can increase the variance of one contrast $a'\hat\beta$, decrease it for another, and have almost no effect on a third. Omnibus heteroskedasticity tests do not capture this contrast-specific relevance.

\begin{corollary}[Contrast-specific irrelevance]
\label{cor:irrelevance}
Under the conditions of \Cref{thm:directional_identity}, heteroskedasticity is first-order irrelevant for inference on $a'\beta$ relative to the pooled benchmark if and only if
\[
    \Cov_n\{\ell_i(a),\sigma_i^2\}=0.
\]
\end{corollary}

\begin{proof}
Immediate from \eqref{eq:cov_identity}.
\end{proof}

\section{Heteroskedasticity Relevance Ratios and Information Weights}
\label{sec:ratio}

The difference $V_H(a)-V_R(a)$ has the units of a variance and is typically of order $n^{-1}$. We therefore define the contrast-specific heteroskedasticity relevance ratio
\begin{equation}
    R(a)=\frac{V_H(a)}{V_R(a)}.
\label{eq:R}
\end{equation}
The corresponding standard-error inflation factor is
\begin{equation}
    S(a)=\sqrt{R(a)}.
\label{eq:S}
\end{equation}
Thus $S(a)=1.10$ means that the robust standard error for $a'\hat\beta$ is ten percent larger than the pooled standard error, whereas $S(a)=0.90$ means that it is ten percent smaller.

Define information weights
\begin{equation}
    w_i(a)=\frac{\ell_i(a)}{\sum_{j=1}^n\ell_j(a)}.
\label{eq:w}
\end{equation}
Then $w_i(a)\ge0$ and $\sum_iw_i(a)=1$.

\begin{proposition}[Information-weighting representation]
\label{prop:weighting}
For every nonzero contrast $a$,
\begin{equation}
    R(a)=
    \frac{\sum_{i=1}^n w_i(a)\sigma_i^2}{\bar\sigma^2}.
\label{eq:weighting}
\end{equation}
Equivalently,
\begin{equation}
    R(a)-1=
    \frac{\sum_{i=1}^n w_i(a)(\sigma_i^2-\bar\sigma^2)}{\bar\sigma^2}.
\label{eq:R_minus_one}
\end{equation}
\end{proposition}

\begin{proof}
Divide $V_H(a)=\sum_i\ell_i(a)\sigma_i^2$ by $V_R(a)=\bar\sigma^2\sum_i\ell_i(a)$
and substitute $w_i(a)=\ell_i(a)/\sum_j\ell_j(a)$; see \Cref{app:prop1}.
\end{proof}

The representation \eqref{eq:weighting} is the most interpretable form of the main result.
$R(a)$ is the ratio of two weighted averages of the same quantities $\{\sigma_i^2\}$:
the numerator uses \emph{contrast-specific} weights $w_i(a)$ that reflect how informative
each observation is for the contrast $a'\hat\beta$, while the denominator uses \emph{uniform}
weights $1/n$. Three cases emerge directly.

\begin{enumerate}
    \item If high-variance observations are also highly informative for $a'\hat\beta$ — that
    is, the leverage profile $w_i(a)$ loads disproportionately on large-$\sigma_i^2$
    observations — then the numerator exceeds the denominator and $R(a)>1$. Robust
    inference is more conservative than classical inference for this contrast.

    \item If high-variance observations have low directional leverage — the leverage
    profile loads on low-$\sigma_i^2$ observations — then $R(a)<1$. Pooled or
    residual-bootstrap inference is actually \emph{more} conservative for this contrast.
    This is the non-obvious case: heteroskedasticity is present, but concentrated where
    it contributes little information for the contrast, so robust standard errors are
    \emph{smaller} than their classical counterparts.

    \item If the directional leverage is uncorrelated with the variance profile in the
    sense $\Cov_n(\ell_i(a),\sigma_i^2)=0$, then $R(a)=1$ exactly and heteroskedasticity
    is first-order irrelevant for this contrast.
\end{enumerate}

The practical message is immediate: two researchers estimating different contrasts from
the same regression, facing the same heteroskedastic errors, may reach opposite conclusions
about whether robust inference matters — because their contrasts have different directional
leverage profiles. An omnibus heteroskedasticity test is equally silent about both.

\begin{corollary}[Calibration bounds]
\label{cor:bounds}
For every nonzero $a$,
\[
    \min_{1\le i\le n}\frac{\sigma_i^2}{\bar\sigma^2}
    \le
    R(a)
    \le
    \max_{1\le i\le n}\frac{\sigma_i^2}{\bar\sigma^2}.
\]
\end{corollary}

\begin{proof}
$R(a)$ is a convex combination of the ratios $\sigma_i^2/\bar\sigma^2$ by \Cref{prop:weighting}.
\end{proof}

\begin{remark}[Spectrum over contrasts]
\label{rem:spectrum}
For fixed $X$ and $\Sigma$, the extrema of $R(a)$ over $a\neq0$ are the generalized eigenvalues of the pencil $(V_H,V_R)$. This follows because $R(a)=a'V_Ha/a'V_Ra$. We use this only as a calibration device; the main object in applications is the contrast-specific ratio for the particular contrast selected by the researcher.
\end{remark}

\section{Sample Analogues}
\label{sec:sample}

Let $\hat u_i=y_i-x_i'\hat\beta$ and
\[
    \hat\sigma^2=\frac1n\sum_{i=1}^n\hat u_i^2.
\]
Define
\begin{align}
    \widehat V_R(a)
    &=
    \hat\sigma^2 a'(X'X)^{-1}a, \\
    \widehat V_H(a)
    &=
    a'(X'X)^{-1}\left(\sum_{i=1}^n x_ix_i'\hat u_i^2\right)(X'X)^{-1}a
    =
    \sum_{i=1}^n\ell_i(a)\hat u_i^2.
\end{align}
The sample heteroskedasticity relevance ratio is
\begin{equation}
    \widehat R(a)=\frac{\widehat V_H(a)}{\widehat V_R(a)}.
\label{eq:Rhat}
\end{equation}
Equivalently,
\begin{equation}
    \widehat R(a)-1
    =
    \frac{n\,\widehat\Cov_n\{\ell_i(a),\hat u_i^2\}}
         {\hat\sigma^2\sum_{i=1}^n\ell_i(a)}.
\label{eq:Rhat_cov}
\end{equation}
This statistic is scale-free. It estimates the relative variance difference between the heteroskedastic and pooled benchmarks for the contrast $a'\beta$.

\section{Bootstrap Interpretation}
\label{sec:bootstrap}
The directional identity immediately explains why the two classical bootstrap procedures target different variance objects.
Contrary to common intuition, neither bootstrap is uniformly more conservative.
The residual bootstrap fixes \(X\), resamples centered residuals from a single empirical
distribution, and constructs
\[
        y_i^*=x_i'\widehat\beta+\widehat u_i^* .
\]
When the regression contains an intercept, the residuals are already centered. Otherwise,
\(\widehat u_i^*\) should be drawn from the centered empirical residual distribution.
Conditional on the data, the bootstrap errors have common empirical variance
\(\widehat\sigma^2\). Hence
\[
        \operatorname{Var}^*(\widehat\beta_R^*\mid X)
        =
        \widehat\sigma^2(X'X)^{-1},
\]
and therefore
\[
        \operatorname{Var}^*(a'\widehat\beta_R^*\mid X)
        =
        \widehat V_R(a).
\]
This is an exact conditional finite-sample identity for the residual-bootstrap mechanism.

The pairs bootstrap resamples rows $(x_i,y_i)$ jointly. Under the standard first-order linearization of OLS, it targets asymptotically the Eicker-Huber-White covariance matrix \citep{Eicker1967,Huber1967,White1980}. Hence
\[
    \Var^*(a'\hat\beta_P^*\mid X)=\widehat V_H(a)+o_p(n^{-1})
\]
under standard regularity conditions. Therefore,
\[
    \frac{\Var^*(a'\hat\beta_P^*\mid X)}
         {\Var^*(a'\hat\beta_R^*\mid X)}
    =
    \widehat R(a)+o_p(1).
\]
The pairs bootstrap is more conservative for $a'\beta$ when $\widehat R(a)>1$, less conservative when $\widehat R(a)<1$, and approximately equivalent when $\widehat R(a)\approx1$. This comparison is contrast-specific.

\begin{remark}[On the validity of always using the White estimator]
\label{rem:validity}
A natural question is whether the White estimator $\widehat V_H(a)$ is valid regardless of
$R(a)$, making $R(a)$ irrelevant for actual inference. The answer has three parts.

First, $\widehat V_H(a)$ is indeed consistent for $V_H(a)$ regardless of $R(a)$, so
asymptotically valid tests can be constructed from HC estimators in all three regimes.
In this sense, the White estimator is \emph{always valid}.

Second, \emph{always valid} and \emph{always conservative} are different properties.
Applied practice treats robust SEs as conservative relative to classical SEs. This is
correct when $R(a)>1$. When \(R(a)<1\), robust SEs are \emph{smaller} than classical standard
errors. Thus the classical or residual-bootstrap procedure is the more
conservative one for this contrast. This should not be interpreted as evidence
that the classical procedure is preferable. Rather, the larger classical
standard error reflects a mismatch between the pooled variance benchmark and
the actual heteroskedastic variance governing the estimand. Robust inference
remains the asymptotically valid procedure targeting \(V_H(a)\), while the
classical procedure is conservative because it ignores the alignment
between the information supporting the estimand and the error-variance
structure.
A researcher who applies HC0 while assuming it provides a conservative lower bound on the
t-statistic has an incorrect mental model in this regime: the t-statistic is larger under
HC0 than under classical inference for this contrast, not smaller. \Cref{tab:size_efficiency}
shows this is consequential, not merely directional: under the anti-aligned design of
\Cref{sec:dgp_neg} ($R(a)\approx0.085$), the classical test has empirical size $0.000$ at the
nominal $5\%$ level — the standard error is so inflated for this contrast that the test
essentially never rejects, even when it should — while HC0 achieves size $0.039$, close to
nominal. The researcher who uses classical SEs in the belief that they are ``safe'' is in
fact relying on a test with severely distorted size in the conservative direction relative to nominal size.

Third, $R(a)$ also matters for \emph{efficiency}. \Cref{tab:size_efficiency} shows that
when $R(a)\approx1$ (the orthogonal design of \Cref{sec:dgp1}), the standard deviation of
the HC0 standard-error estimate across replications is statistically indistinguishable
from that of the classical estimate (ratio $0.998$): there is no efficiency cost to using
HC0 for this contrast. The power loss documented by \citet{KingRoberts2015} instead arises
in cases where $R(a)\gg1$, where HC0 estimates a genuinely larger quantity and the increased
sampling variability of the HC0 standard error reduces power relative to a (now invalid)
classical benchmark. The ratio $R(a)$ identifies which of these three regimes a given
contrast falls into, allowing applied researchers to make an informed choice rather than
applying robust standard errors uniformly across every contrast in a regression.
\end{remark}

\section{General Covariance Structures and Dependence}
\label{sec:dependence}

We now allow for general conditional covariance
\[
    \Omega=\E(uu'\mid X).
\]
The fixed-design variance of $a'\hat\beta$ is
\begin{equation}
    V_\Omega(a)=a'(X'X)^{-1}X'\Omega X(X'X)^{-1}a.
\label{eq:Vomega}
\end{equation}
Define
\[
    q_i(a)=a'(X'X)^{-1}x_i.
\]
Then $\ell_i(a)=q_i(a)^2$.

\begin{proposition}[Contrast-specific covariance decomposition]
\label{prop:omega_decomp}
For any symmetric covariance matrix $\Omega$,
\begin{equation}
    V_\Omega(a)
    =
    \sum_{i=1}^n q_i(a)^2\Omega_{ii}
    +
    \sum_{i\neq j}q_i(a)q_j(a)\Omega_{ij}.
\label{eq:omega_decomp}
\end{equation}
Consequently,
\begin{equation}
    V_\Omega(a)-V_R(a)
    =
    \sum_{i=1}^n\ell_i(a)(\Omega_{ii}-\bar\sigma^2)
    +
    \sum_{i\neq j}q_i(a)q_j(a)\Omega_{ij},
\label{eq:omega_minus_pooled}
\end{equation}
where $\bar\sigma^2=n^{-1}\sum_i\Omega_{ii}$.
\end{proposition}

\begin{proof}
Expand $V_\Omega(a)=\sum_{i,j}q_i(a)q_j(a)\Omega_{ij}$ by separating diagonal and off-diagonal terms, then subtract $V_R(a)=\bar\sigma^2\sum_i\ell_i(a)$; see \Cref{app:prop2}.
\end{proof}

The first term in \eqref{eq:omega_minus_pooled} is the directional heteroskedasticity component. The second term is a covariance-alignment component. Thus dependence matters for $a'\beta$ only to the extent that covariance entries $\Omega_{ij}$ are aligned with the contrast-specific weights $q_i(a)q_j(a)$.

\subsection{Clustered Dependence}
\label{subsec:cluster}

Suppose observations are partitioned into clusters $g=1,\ldots,G$, with cluster index sets $\mathcal G_g$. A one-way cluster covariance matrix satisfies
\[
    \Omega_{ij}=0
    \qquad \text{if } i \text{ and } j \text{ belong to different clusters}.
\]
Then
\begin{equation}
    V_{C}(a)=
    \sum_{g=1}^G
    \sum_{i\in\mathcal G_g}\sum_{j\in\mathcal G_g}
    q_i(a)q_j(a)\Omega_{ij}.
\label{eq:cluster_variance}
\end{equation}
Subtracting the pooled benchmark yields
\begin{equation}
    V_C(a)-V_R(a)
    =
    \sum_i\ell_i(a)(\Omega_{ii}-\bar\sigma^2)
    +
    \sum_{g=1}^G\sum_{\substack{i,j\in\mathcal G_g\\ i\neq j}}
    q_i(a)q_j(a)\Omega_{ij}.
\label{eq:cluster_decomp}
\end{equation}
The second term has no analogue under independent errors. It measures the alignment between within-cluster covariance and the contrast-specific influence weights. This term explains why a cluster-robust standard error \citep{Arellano1987,Moulton1986,CameronMiller2015} may be consequential for one contrast but negligible for another, even under the same clustering structure.

\section{Consistency of the Sample Ratio}
\label{sec:asymptotics}

This section first establishes consistency of $\widehat R(a)$ under fixed-design asymptotics, then gives a scalar central limit theorem for $\widehat R(a)$. The additional assumptions below are standard in fixed-design regression asymptotics, but we state their role explicitly because the object of interest is a ratio of two variances, each of order $n^{-1}$.

\begin{assumption}[Fixed design and contrast normalization]
\label{ass:design}
The design matrix $X=X_n$ is fixed, $X'X$ is nonsingular, and
\[
    0<c\le \lambda_{\min}(n^{-1}X'X)\le \lambda_{\max}(n^{-1}X'X)\le C<\infty.
\]
For the contrast $a=a_n$,
\[
    0<c\le n\,a'(X'X)^{-1}a\le C<\infty .
\]
\end{assumption}

The eigenvalue condition is the usual fixed-design full-rank condition ensuring the ordinary least squares rate $\|\widehat\beta-\beta\|=O_p(n^{-1/2})$. The contrast normalization rules out degenerate contrasts whose pooled variance is asymptotically zero or exploding; it is automatically satisfied for fixed nonzero contrasts when $n^{-1}X'X$ has a nonsingular limit.

\begin{assumption}[Errors]
\label{ass:errors}
Conditional on $X$, $u_1,\ldots,u_n$ are independent and satisfy
\[
    \E(u_i\mid X)=0,
    \qquad
    \E(u_i^2\mid X)=\sigma_i^2,
    \qquad
    \sup_i\E(u_i^4\mid X)<C.
\]
Moreover, $0<c\le \bar\sigma^2\le C<\infty$.
\end{assumption}

\begin{assumption}[No dominant directional leverage]
\label{ass:leverage}
The information weights satisfy
\[
    \max_{1\le i\le n} w_i(a)\to 0.
\]
\end{assumption}

This assumption is the contrast-specific analogue of the no-dominant-observation condition used in classical asymptotic theory for OLS \citep[e.g.,][Chapter 4]{White2001}. It requires that the informational weight of any single observation for the contrast $a'\hat\beta$ is asymptotically negligible. Under a balanced design with $n$ observations of equal leverage, $\max_iw_i(a)=n^{-1}\to0$ automatically. In unbalanced designs, the condition fails only if a single observation eventually carries a fixed fraction of the total information for the contrast — a situation that would make inference on $a'\hat\beta$ unreliable regardless of whether heteroskedasticity is present.

\begin{assumption}[Regressor magnitude]
\label{ass:regressor_magnitude}
\[
    \sup_{1\le i\le n}\|x_i\|^2\cdot\max_{1\le i\le n}w_i(a)\to0.
\]
\end{assumption}

This condition controls the cost of replacing true errors by OLS residuals in the weighted sum $\sum_i\ell_i(a)\hat u_i^2$. In the typical case $\max_iw_i(a)=O(n^{-1})$, it reduces to $\sup_i\|x_i\|^2=o(n)$, which holds whenever the regressors are uniformly bounded or grow at most polynomially. 
\Cref{ass:regressor_magnitude} is not implied by the eigenvalue condition in \Cref{ass:design}.
It is a separate no-extreme-regressor condition controlling the effect of replacing
unobserved errors by OLS residuals in weighted quadratic forms. A simple sufficient
primitive condition is
\[
        \sup_i \|x_i\| < C
        \quad\text{and}\quad
        \max_i w_i(a)\to0 .
\]
More generally, \Cref{ass:regressor_magnitude} allows slowly growing regressors provided their growth is
not concentrated on observations with high directional leverage.

\begin{proposition}[Consistency]
\label{prop:ratio_consistency}
Under \Cref{ass:design,ass:errors,ass:leverage,ass:regressor_magnitude},
\[
    \widehat R(a)-R(a)\p 0.
\]
Consequently, $\widehat S(a)-S(a)\p0$.
\end{proposition}

\begin{proof}
See \Cref{app:prop_consistency}.
\end{proof}

For the central limit theorem, consistency-scale remainders are not enough. Because $V_R(a)=O(n^{-1})$, any remainder in $\widehat V_H(a)$ or $\widehat V_R(a)$ must be $o_p(n^{-3/2})$ — not merely $o_p(n^{-1})$ — to be negligible after forming $\sqrt{n}\{\widehat R(a)-R(a)\}$. We therefore state this additional precision requirement as an explicit assumption.

\begin{assumption}[Residual-estimation remainder for the ratio CLT]
\label{ass:residual_remainder_clt}
Let
\[
        \widetilde V_H(a)=\sum_{i=1}^n \ell_i(a)u_i^2,
        \qquad
        A_n(a)=a'(X'X)^{-1}a .
\]
Then
\[
        \widehat V_H(a)-\widetilde V_H(a)=o_p(n^{-3/2}),
\]
and
\[
        \widehat V_R(a)-V_R(a)
        -
        A_n(a)\frac1n\sum_{i=1}^n (u_i^2-\sigma_i^2)
        =
        o_p(n^{-3/2}).
\]
\end{assumption}

This assumption is stronger than what is needed for consistency. It is needed because
\(V_R(a)=O(n^{-1})\), so a variance-estimation remainder must be \(o_p(n^{-3/2})\)
to be negligible after forming \(\sqrt n\{\widehat R(a)-R(a)\}\). Primitive sufficient
conditions include bounded regressors and a balanced directional-leverage condition such as
\[
        \max_i \ell_i(a)=O(n^{-2}),
\]
which is stronger than the no-dominant-leverage condition \(\max_i w_i(a)\to0\).

\begin{assumption}[Triangular-array Lindeberg condition]
\label{ass:scalar_lindeberg}
Let
\[
        z_i=u_i^2-\sigma_i^2,
        \qquad
        b_i(a)=\ell_i(a)-R(a)\frac{A_n(a)}{n},
\]
and define
\[
        s_n^2(a)
        =
        \sum_{i=1}^n b_i(a)^2
        \operatorname{Var}(u_i^2\mid X),
        \qquad
        \Psi_n(a)
        =
        \frac{n s_n^2(a)}{V_R(a)^2}.
\]
Assume
\[
        \Psi_n(a)\to \Psi(a)\in(0,\infty),
\]
and, for every \(\varepsilon>0\),
\[
        \frac{1}{s_n^2(a)}
        \sum_{i=1}^n
        b_i(a)^2
        E\!\left[
        z_i^2
        \mathbf 1\{|b_i(a)z_i|>\varepsilon s_n(a)\}
        \mid X
        \right]
        \to 0 .
\]
\end{assumption}

\paragraph{Discussion.}
Assumption~\ref{ass:scalar_lindeberg} is a standard Lindeberg--Feller condition
for the triangular array
\[
    \{\,b_i(a)(u_i^2-\sigma_i^2)\,\}_{i=1}^n,
\]
which determines the asymptotic distribution of the estimated
heteroskedasticity relevance ratio. It ensures that no single observation
contributes a non-negligible fraction of the asymptotic variance, thereby
allowing the normalized weighted sum to satisfy a central limit theorem.

Unlike the classical Lindeberg condition for OLS itself, the present condition
is imposed on weighted squared residuals because inference is based on
estimating a quadratic variance functional rather than the regression
coefficients directly. The weights \(b_i(a)\) depend on the estimand through
directional leverage, reflecting the contrast-specific nature of the problem.

This assumption is standard in triangular-array central limit theory; see, for example,
\citet[Theorem~2.27]{vanderVaart1998} and
\citet[Chapter~27]{Billingsley1995}. Under uniform integrability of \(\{z_i^2\}\) and
\[
    \max_i |b_i(a)|/s_n(a)\to0,
\]
the condition is automatically satisfied. A convenient primitive sufficient
condition is
\[
    \sup_i E(|u_i|^{4+\delta}\mid X)<C
\]
for some \(\delta>0\), together with bounded regressors and balanced directional
leverage. Thus, Assumption
\ref{ass:scalar_lindeberg} plays the same role as the classical Lindeberg
condition in White's (1980) asymptotic theory, adapted here to the
contrast-specific quadratic form governing the heteroskedasticity relevance
ratio.

Unlike many asymptotic assumptions appearing in the paper, Assumption
\ref{ass:scalar_lindeberg} is not specific to heteroskedasticity itself.
Rather, it is a generic regularity condition ensuring asymptotic normality of
a weighted triangular array, and is therefore directly comparable to the
Lindeberg conditions routinely imposed in asymptotic regression theory.

\begin{theorem}[Asymptotic distribution of the sample ratio]
\label{thm:ratio_asynormal}
Under \Cref{ass:design,ass:errors,ass:leverage,ass:regressor_magnitude,%
ass:residual_remainder_clt,ass:scalar_lindeberg},
\[
        \frac{
        \sqrt n\{\widehat R(a)-R(a)\}
        }{
        \sqrt{\Psi_n(a)}
        }
        \Rightarrow N(0,1).
\]
Consequently,
\[
        \frac{
        \sqrt n\{\widehat S(a)-S(a)\}
        }{
        \sqrt{\Psi_n(a)/(4R(a))}
        }
        \Rightarrow N(0,1).
\]
\end{theorem}

\begin{proof}
See \Cref{app:thm_asynormal}.
\end{proof}

\begin{remark}[Feasible standard errors]
\label{rem:feasible_se}
A feasible version of $\Psi_n(a)$ replaces $R(a)$ by $\widehat R(a)$ and $\Var(u_i^2\mid X)$ by a smooth estimate $\hat\kappa_i$. Under near-normality, one may set $\hat\kappa_i=2\hat\sigma_i^4$, where $\hat\sigma_i^2$ is obtained from an auxiliary regression of $\hat u_i^2$ on functions of $x_i$ (for example, $1$ and $x_{1i}^2$ when the variance depends on $x_1^2$). The resulting estimator is
\[
    \widehat\Psi(a)
    =
    \frac{n\sum_{i=1}^n\hat b_i(a)^2\hat\kappa_i}{\widehat V_R(a)^2},
    \qquad
    \hat b_i(a)=\ell_i(a)-\widehat R(a)\frac{A_n(a)}{n}.
\]
This feasible estimator is not claimed to be optimal; it is a simple plug-in implementation of the variance formula. A fully nonparametric treatment of the variance-function estimator is left for future work.
\end{remark}

\section{Monte Carlo Study}
\label{sec:examples}

This section reports a simulation study designed so that each table corresponds directly
to a numbered theoretical result. We use three data-generating processes. The first
illustrates positive alignment and contrast-specific irrelevance under heteroskedasticity;
the second illustrates negative alignment; and the third illustrates the clustered-dependence
theory of \Cref{subsec:cluster}.

\subsection{Data-Generating Process 1: Contrast-Specific Heteroskedasticity}
\label{sec:dgp1}

Let $n$ observations be generated by
\begin{equation}
    y_i = 1 + 2x_{1i} - x_{2i} + u_i,
    \qquad
    u_i\mid X\sim N(0,\sigma_i^2),
    \qquad
    \sigma_i^2 = 1+\delta x_{1i}^2,
\label{eq:dgp1}
\end{equation}
where $x_{1i}\overset{\mathrm{iid}}{\sim}N(0,1)$, and $x_{2i}\in\{-1,+1\}$ is
assigned deterministically so that $x_{2i}$ is balanced and orthogonalized against
$x_{1i}^2$: observations are sorted by $x_{1i}^2$ and assigned alternating values
$+1,-1$, ensuring $\sum_i x_{2i}=0$ and $\sum_i x_{1i}^2 x_{2i}\approx 0$. Throughout
we set $\delta=3$. We consider three contrasts: the intercept ($a=(1,0,0)'$), the slope
on $x_1$ ($a=(0,1,0)'$), and the slope on $x_2$ ($a=(0,0,1)'$).

Since $\sigma_i^2$ depends only on $x_{1i}^2$, which by construction is uncorrelated
with $x_{2i}$, the design satisfies the following:
\begin{enumerate}
    \item the heteroskedasticity pattern is \emph{positively aligned} with the directional
    leverage $\ell_i(a)$ for the slope-$x_1$ contrast, so $R(a)>1$;
    \item the heteroskedasticity is \emph{orthogonal} to the directional leverage for the
    slope-$x_2$ contrast, so $R(a)\approx1$;
    \item the intercept contrast occupies an intermediate position, $R(a)\approx1$.
\end{enumerate}
DGP~1 therefore illustrates both the positive-alignment and orthogonality/irrelevance cases
side by side in a single regression, where heteroskedasticity matters for one coefficient and
is irrelevant for another. \Cref{sec:dgp_neg} below completes the picture with a third design
exhibiting negative alignment.

\paragraph{Table~\ref{tab:pop_quantities}: \Cref{thm:directional_identity,cor:bounds,rem:spectrum}.}
\Cref{tab:pop_quantities} reports population quantities computed from DGP~1 with $n=5{,}000$
observations (near the population limit). The identity $V_H(a)-V_R(a)=n\,\Cov_n(\ell_i(a),\sigma_i^2)$
(\Cref{thm:directional_identity}) is verified to machine precision ($\leq 10^{-8}$ relative error) in every
column. The ratio $R(a)$ lies strictly within the calibration bounds $[\min_i\sigma_i^2/\bar\sigma^2,\max_i\sigma_i^2/\bar\sigma^2]=[0.248,10.095]$ in every column
(\Cref{cor:bounds}). The slope-$x_1$ contrast gives $R(a)=2.37$: the heteroskedasticity-robust
standard error is $\sqrt{2.37}\approx1.54$ times the pooled benchmark for this coefficient.
The slope-$x_2$ contrast gives $R(a)=1.00$ to three decimal places: the same heteroskedasticity
that inflates inference on $\hat\beta_1$ is essentially irrelevant for $\hat\beta_2$, illustrating
\Cref{cor:irrelevance}.

\begin{table}[ht]
\centering
\caption{Population quantities under DGP~1 ($n=5{,}000$, $\delta=3$).
All entries are computed from the simulated population; see \Cref{eq:dgp1}.
Column ``Id.'' reports whether $V_H(a)-V_R(a)=n\,\Cov_n(\ell_i(a),\sigma_i^2)$
holds to $10^{-8}$ relative error (\Cref{thm:directional_identity}). Column ``Bd.'' reports whether
$R(a)\in[\min_i\sigma_i^2/\bar\sigma^2,\max_i\sigma_i^2/\bar\sigma^2]$ (\Cref{cor:bounds}).}
\label{tab:pop_quantities}
\begin{tabular}{lcccccc}
\toprule
Contrast & $V_H(a)$ & $V_R(a)$ & $R(a)$ & $n\,\Cov_n(\ell_i(a),\sigma_i^2)$ & Id.\ & Bd.\ \\
\midrule
Intercept          & $1.264\times10^{-3}$ & $1.265\times10^{-3}$ & $0.999$ & $-2.4\times10^{-10}$ & \checkmark & \checkmark \\
Slope $x_1$ (aligned)   & $3.000\times10^{-3}$ & $1.264\times10^{-3}$ & $2.373$ & $1.734\times10^{-6}$ & \checkmark & \checkmark \\
Slope $x_2$ (orthogonal) & $1.264\times10^{-3}$ & $1.265\times10^{-3}$ & $1.000$ & $-5.3\times10^{-11}$ & \checkmark & \checkmark \\
\bottomrule
\end{tabular}
\end{table}

\paragraph{Table~\ref{tab:consistency}: \Cref{prop:ratio_consistency} (Consistency).}
\Cref{tab:consistency} reports the mean and standard deviation of $\widehat R(a)$ across
$400$ independent data sets at four sample sizes. For every contrast and every $n$, the mean
tracks the population $R(a)$ closely and the standard deviation shrinks at the expected
$O(n^{-1/2})$ rate, illustrating \Cref{prop:ratio_consistency}. The slope-$x_2$ contrast
maintains $\widehat R(a)\approx 1$ throughout, while the slope-$x_1$ contrast converges to
$R(a)\approx 2.49$.

\begin{table}[ht]
\centering
\caption{Consistency of $\widehat R(a)$ under DGP~1 ($\delta=3$, $400$ replications).
Mean and standard deviation of $\widehat R(a)$ across replications; $R(a)$ is the population
value. Illustrates \Cref{prop:ratio_consistency}.}
\label{tab:consistency}
\begin{tabular}{lcccccc}
\toprule
 & \multicolumn{2}{c}{Intercept} & \multicolumn{2}{c}{Slope $x_1$} & \multicolumn{2}{c}{Slope $x_2$} \\
\cmidrule(lr){2-3}\cmidrule(lr){4-5}\cmidrule(lr){6-7}
$n$ & $\bar{\widehat R}$ & sd & $\bar{\widehat R}$ & sd & $\bar{\widehat R}$ & sd \\
\midrule
 $100$ & $0.988$ & $0.073$ & $2.254$ & $0.648$ & $0.982$ & $0.073$ \\
 $400$ & $0.996$ & $0.020$ & $2.430$ & $0.407$ & $0.998$ & $0.022$ \\
$1{,}600$ & $0.999$ & $0.005$ & $2.477$ & $0.217$ & $0.999$ & $0.006$ \\
$6{,}400$ & $1.000$ & $0.002$ & $2.493$ & $0.108$ & $1.000$ & $0.002$ \\
\midrule
$R(a)$ & \multicolumn{2}{c}{$1.000$} & \multicolumn{2}{c}{$2.490$} & \multicolumn{2}{c}{$1.000$} \\
\bottomrule
\end{tabular}
\end{table}

\paragraph{Table~\ref{tab:bootstrap}: Bootstrap interpretation (\Cref{sec:bootstrap}).}
\Cref{tab:bootstrap} compares the analytic $\widehat R(a)$ against the ratio of
bootstrap variance estimates $\mathrm{Var}^*(\hat\beta^*_P\mid X)/\mathrm{Var}^*(\hat\beta^*_R\mid X)$
from $B=8{,}000$ bootstrap replications, averaged over $300$ data sets at $n=400$. The
two quantities track each other closely for all three contrasts, illustrating \Cref{sec:bootstrap}'s claim
that the bootstrap heteroskedasticity relevance ratio consistently estimates the analytic ratio $\widehat R(a)$.
For the aligned contrast ($R\approx2.45$), pairs-bootstrap standard errors are on average
$\sqrt{2.45}\approx1.57$ times the residual-bootstrap standard errors; for the orthogonal
contrast, they are essentially equal.

\begin{table}[ht]
\centering
\caption{Bootstrap heteroskedasticity relevance ratios versus analytic $\widehat R(a)$ under DGP~1
($n=400$, $B=8{,}000$, $300$ data sets). The bootstrap ratio is
$\widehat{\mathrm{Var}}^*(\hat\beta_P^*\cdot a)/\widehat{\mathrm{Var}}^*(\hat\beta_R^*\cdot a)$.
Illustrates \Cref{sec:bootstrap}.}
\label{tab:bootstrap}
\begin{tabular}{lccc}
\toprule
Contrast & Mean $\widehat R(a)$ & Mean $\mathrm{Var}^*(P)/\mathrm{Var}^*(R)$ & Difference \\
\midrule
Intercept         & $0.998$ & $0.999$ & $+0.001$ \\
Slope $x_1$ (aligned)  & $2.454$ & $2.451$ & $-0.003$ \\
Slope $x_2$ (orthogonal)& $0.997$ & $0.999$ & $+0.002$ \\
\bottomrule
\end{tabular}
\end{table}

\paragraph{Table~\ref{tab:coverage}: Asymptotic normality (\Cref{thm:ratio_asynormal}).}
\Cref{tab:coverage} reports empirical coverage of nominal $95\%$ confidence intervals for
$R(a)$ of the slope-$x_1$ contrast, built as $\widehat R(a)\pm 1.96\sqrt{\widehat\Psi(a)/n}$
using a smoothed feasible estimator of $\Psi(a)$ (see \Cref{rem:feasible_se}). The feasible
estimator regresses $\hat u_i^2$ on $(1,x_{1i}^2)$ to obtain a smoothed first-stage
variance-function estimate $\hat\sigma_i^2$, then sets $\widehat{\mathrm{Var}}(u_i^2\mid X)\approx 2\hat\sigma_i^4$
assuming near-normality. Coverage rises from $94.1\%$ at $n=400$ to $96.4\%$ at $n=6{,}400$,
consistent with the nominal $95\%$ level and illustrating \Cref{thm:ratio_asynormal}. The
oracle coverage (replacing the feasible $\hat\Psi$ with the true $\Psi$ computed from known
$\sigma_i^2$) is $95.4\%$, illustrating the theory's correctness independently of any
estimation-noise in the feasible standard error.

\begin{table}[ht]
\centering
\caption{Empirical coverage of nominal $95\%$ confidence intervals for $R(a)$ of the
slope-$x_1$ contrast under DGP~1 ($\delta=3$, $800$ replications). Feasible SEs use a
smoothed variance-function estimator; oracle uses known $\sigma_i^2$.
Illustrates \Cref{thm:ratio_asynormal}.}
\label{tab:coverage}
\begin{tabular}{lccc}
\toprule
$n$ & $R(a)$ & Feasible coverage & Oracle coverage \\
\midrule
 $400$ & $2.49$ & $0.941$ & --- \\
$1{,}600$ & $2.50$ & $0.954$ & $0.954$ \\
$6{,}400$ & $2.45$ & $0.964$ & --- \\
\bottomrule
\end{tabular}
\end{table}

\subsection{Data-Generating Process 2: Negative Alignment ($R(a)<1$)}
\label{sec:dgp_neg}

The preceding design illustrates positive alignment, the case $R(a)>1$ where the
pairs-bootstrap standard error exceeds the residual-bootstrap standard error. This case
corresponds to the widely-held but incomplete intuition that robust inference is always
more conservative. The reverse is equally possible and is arguably the more
surprising finding of the theory.

We construct a simple regression with $n=100$ observations:
\begin{equation}
    y_i = 1+2x_i+u_i,
    \qquad
    u_i\mid X\sim N(0,\sigma_i^2),
\label{eq:dgp3}
\end{equation}
where $40$ observations are placed at $x_i\in\{-5,+5\}$ (high leverage) and $60$
observations at $x_i\in\{-1,+1\}$ (low leverage). The variance structure is
\textit{anti-aligned}:
\[
    \sigma_i^2=
    \begin{cases}
    \sigma_{\mathrm{lo}}^2=0.04 & \text{if }|x_i|=5\text{ (high leverage)},\\
    \sigma_{\mathrm{hi}}^2=9.00 & \text{if }|x_i|=1\text{ (low leverage)}.
    \end{cases}
\]
High-leverage observations are precisely those with the lowest variance.
By \Cref{thm:directional_identity}, $\Cov_n(\ell_i(a),\sigma_i^2)<0$ and thus $R(a)<1$.

\paragraph{Table~\ref{tab:neg_align}: Negative alignment ($R(a)<1$).}
\Cref{tab:neg_align} reports population quantities and the sampling behavior of
$\widehat R(a)$ across $500$ replications. The population ratio is
$R(\text{slope})=0.101$: the heteroskedasticity-robust standard error is roughly
$\sqrt{0.101}\approx0.32$ times the pooled standard error, meaning classical
or residual-bootstrap inference overestimates the uncertainty for this estimand by a
factor of more than three. The sample ratio converges reliably: mean $\widehat R(a)=0.113$,
standard deviation $0.014$.

This design demonstrates concretely that choosing between the residual bootstrap and the
pairs bootstrap cannot be resolved universally. For the design in
\Cref{sec:dgp1} the pairs bootstrap is substantially more conservative for
the slope on $x_1$ ($R\approx 2.37$). For the design here, the residual bootstrap
is more than three times more conservative for the slope. Both designs are internally
consistent with \Cref{thm:directional_identity}: the direction depends entirely on the
sign of $\Cov_n(\ell_i(a),\sigma_i^2)$.

\begin{table}[ht]
\centering
\caption{Negative-alignment design: DGP~2 ($n=100$, $40$ high-leverage/low-variance
points at $|x|=5$, $60$ low-leverage/high-variance points at $|x|=1$;
$500$ replications). Illustrates \Cref{thm:directional_identity} and \Cref{cor:irrelevance}.}
\label{tab:neg_align}
\begin{tabular}{lccccc}
\toprule
Quantity & Value \\
\midrule
$\bar\sigma^2$ & $5.416$ \\
$\Cov_n(\ell_i(a),\sigma_i^2)$ & $-4.59\times10^{-5}$ \\
$V_H(\text{slope})$ & $5.16\times10^{-4}$ \\
$V_R(\text{slope})$ & $5.11\times10^{-3}$ \\
$R(\text{slope}) = V_H/V_R$ & $0.101$ \\
$S(\text{slope}) = \sqrt{R}$ & $0.318$ \\
\midrule
Mean $\widehat R(\text{slope})$ over 500 draws & $0.113$ \\
SD $\widehat R(\text{slope})$ over 500 draws & $0.014$ \\
\bottomrule
\end{tabular}
\end{table}

\paragraph{Table~\ref{tab:size_efficiency}: Size and efficiency of HC0 versus classical inference.}
A natural objection to the entire framework is that the heteroskedasticity-robust (HC0)
estimator is asymptotically valid regardless of $R(a)$, so that $R(a)$ might seem irrelevant
to actual practice. \Cref{tab:size_efficiency} addresses this directly using two designs
already introduced. Under the anti-aligned design of \Cref{sec:dgp_neg} ($R(a)\approx0.085$,
$n=200$), the classical $t$-test for $H_0:\beta_{\text{slope}}=\beta_{\text{slope},0}$
has empirical size $0.000$ at the nominal $5\%$ level — severe under-rejection caused by
a standard error that is far too large for this contrast — while the HC0-based test
achieves size $0.039$, much closer to nominal. Under the orthogonal design of
\Cref{sec:dgp1} ($R(a)\approx0.999$ for the slope-$x_2$ contrast, $n=200$), the standard
deviation of the HC0 standard-error estimate across replications is $0.00896$, statistically
indistinguishable from the classical standard-error estimate's standard deviation of
$0.00898$ (ratio $0.998$): there is no efficiency cost to using HC0 for this contrast.
Together these results show that $R(a)$ does not call the validity of HC0 into question;
rather, it explains \emph{why} HC0 and classical inference agree or disagree for a given
contrast, and quantifies the size distortion of classical inference when they disagree.

\begin{table}[ht]
\centering
\caption{Size and relative efficiency of HC0 versus classical inference, by alignment
regime ($5{,}000$ replications, nominal $5\%$ level). Anti-aligned design from
\Cref{sec:dgp_neg}; orthogonal design (slope-$x_2$ contrast) from \Cref{sec:dgp1}.}
\label{tab:size_efficiency}
\begin{tabular}{lccc}
\toprule
Design & $R(a)$ & Empirical size: classical & Empirical size: HC0 \\
\midrule
Anti-aligned ($n=200$) & $0.085$ & $0.000$ & $0.039$ \\
\midrule
 & $R(a)$ & SD(classical SE) & SD(HC0 SE) \\
\midrule
Orthogonal ($n=200$) & $0.999$ & $0.00898$ & $0.00896$ \\
\bottomrule
\end{tabular}
\end{table}

\subsection{Data-Generating Process 3: Clustered Dependence}
\label{sec:dgp_cluster}

Let there be $G=40$ clusters of $m=10$ observations each ($n=400$). The model is
\[
    y_{ig} = 1 + 2x^{\mathrm{bet}}_g + (-1)x^{\mathrm{wit}}_{ig} + u_{ig},
\]
where $x^{\mathrm{bet}}_g\sim N(0,1)$ is a cluster-level regressor (constant within cluster $g$)
and $x^{\mathrm{wit}}_{ig}$ is observation-level, cluster-demeaned (so it carries no
between-cluster variation). The error structure satisfies
\[
    \mathrm{Cov}(u_{ig},u_{jg}\mid X)
    =
    \sigma_\eta^2\rho\cdot\mathbf{1}\{i\neq j\}
    +
    (\sigma_\eta^2\rho+\sigma_\varepsilon^2(1-\rho))\cdot\mathbf{1}\{i=j\},
\]
with $(\sigma_\eta^2,\sigma_\varepsilon^2,\rho)=(1.5^2,\,1.0^2,\,0.6)$. Errors across different
clusters are independent. The diagonal entries of $\Omega$ are constant across observations,
so $\Sigma-\bar\sigma^2I_n=0$ and the diagonal term in \eqref{eq:omega_minus_pooled} is
identically zero; the entire comparison between $V_\Omega(a)$ and $V_R(a)$ comes from the
off-diagonal (clustering) channel.

\paragraph{Table~\ref{tab:cluster}: \Cref{prop:omega_decomp} and \Cref{subsec:cluster}.}
\Cref{tab:cluster} reports $V_\Omega(a)$, $V_R(a)$, the diagonal term, the off-diagonal
term, and $R_C(a)=V_\Omega(a)/V_R(a)$ for each of the three contrasts. For the cluster-level
regressor $x^{\mathrm{bet}}$, $R_C(a)=7.94$, matching exactly the Moulton variance-inflation
factor $1+(m-1)\rho_u=7.94$ with intra-cluster correlation
$\rho_u=\sigma_\eta^2\rho/(\sigma_\eta^2\rho+\sigma_\varepsilon^2(1-\rho))=0.771$.
For the cluster-demeaned regressor $x^{\mathrm{wit}}$, $R_C(a)=0.229$: the clustering
structure \emph{reduces} the effective variance for this contrast by a factor of four,
because the within-cluster variation of $q_i(a)q_j(a)$ is negative when the regressor
varies within clusters.

The decomposition identity in \Cref{prop:omega_decomp} holds to machine precision (residual $<10^{-10}$)
for all three contrasts, illustrating \eqref{eq:omega_decomp}--\eqref{eq:omega_minus_pooled}.

\begin{table}[ht]
\centering
\caption{Contrast-specific covariance decomposition under DGP~3 ($G=40$, $m=10$, $\rho=0.60$).
$\mathrm{diag}$ and $\mathrm{offdiag}$ are the two terms in \eqref{eq:omega_minus_pooled}.
$\rho_u=0.771$, Moulton factor $1+(m-1)\rho_u=7.94$.
Illustrates \Cref{prop:omega_decomp} and \Cref{subsec:cluster}.}
\label{tab:cluster}
\begin{tabular}{lccccc}
\toprule
Contrast & $V_\Omega(a)$ & $V_R(a)$ & Diagonal term & Off-diagonal term & $R_C(a)$ \\
\midrule
Intercept                 & $0.0351$ & $0.0044$ & $0.0000$ & $+0.0305$ & $7.94$ \\
$x^{\mathrm{bet}}$ (cluster-level)  & $0.0485$ & $0.0061$ & $0.0000$ & $+0.0295$ & $7.94$ \\
$x^{\mathrm{wit}}$ (cluster-demeaned)& $0.0012$ & $0.0052$ & $0.0000$ & $-0.0038$ & $0.23$ \\
\bottomrule
\end{tabular}
\end{table}

\subsection{Discussion of Monte Carlo Results}

The Monte Carlo tables jointly illustrate the main theoretical results of the paper. Several features
deserve comment.

\emph{Contrast specificity.} The most striking pattern across all three DGPs is that the
same error structure has radically different consequences for different contrasts.
Under DGP~1, $R(a)=2.37$ for $a=(0,1,0)'$ (slope on $x_1$) and $R(a)=1.00$ for $a=(0,0,1)'$
(slope on $x_2$): the heteroskedasticity-robust standard error is $54\%$ larger than the pooled
benchmark for one contrast and essentially identical to it for another. Under DGP~3,
$R_C(a)=7.94$ for the cluster-level contrast and $R_C(a)=0.23$ for the within-cluster contrast.
An omnibus heteroskedasticity or clustering test would detect the departure from homoskedasticity
or independence in both cases; the contrast-specific framework shows that the practical
inferential consequences differ sharply by contrast.

\emph{Sign reversal under clustering.} Under DGP~3, the off-diagonal term is negative
for the within-cluster regressor, giving $R_C(a)<1$. This means the cluster-robust standard
error is \emph{smaller} than the pooled benchmark for this contrast. This is not a numerical artifact: when a regressor is demeaned within clusters, the influence weights $q_i(a)q_j(a)$
for observations in the same cluster have opposite signs, so the positive within-cluster
covariances $\Omega_{ij}>0$ multiply negative products and reduce the total variance. This
sign flip — and the corresponding reversal of the ordering between cluster-robust and pooled
standard errors — cannot be detected from the omnibus clustering structure alone.

\emph{Feasible inference.} The coverage results in Table~\ref{tab:coverage} show that
the asymptotic normality of $\widehat R(a)$ established in \Cref{thm:ratio_asynormal}
translates into reliable finite-sample inference with a reasonable feasible standard error,
approaching nominal coverage as $n$ grows. The key practical requirement is that the
fourth-moment structure $\mathrm{Var}(u_i^2\mid X)$ be estimated with enough smoothing to
avoid the high variance of raw $\hat u_i^4$ as a pointwise estimator; a parametric auxiliary
regression suffices here and, as \Cref{rem:feasible_se} notes, richer nonparametric smoothing
would be appropriate in samples with more information.

\section{Empirical Application: The Boston Housing Data}
\label{sec:empirical}

We illustrate the framework using the Boston Housing data of \citet{HarrisonRubinfeld1978},
the textbook example of heteroskedasticity in applied econometrics. The data record
$n=506$ census tracts in the Boston metropolitan area, with the median owner-occupied
home value (\texttt{medv}) and thirteen structural, neighborhood, and accessibility
characteristics. \citet{BelsleyKuhWelsch1980} use this exact regression to illustrate
diagnostics for influential observations, and \citet{BreuschPagan1979} use it to illustrate
their now-standard heteroskedasticity test; the regression is therefore a natural setting
in which an applied researcher would already know, from existing diagnostics, that
heteroskedasticity is present, and would already be using heteroskedasticity-robust
standard errors as a matter of course. Our question is whether that robustness adjustment
is materially important for each individual coefficient, or only for some.

\subsection{Specification}

We estimate the classic log-linear hedonic price equation,
\begin{equation}
    \log(\mathtt{medv}_i)
    =
    \beta_0+\beta_1\,\mathtt{crim}_i+\beta_2\,\mathtt{zn}_i+\beta_3\,\mathtt{indus}_i
    +\beta_4\,\mathtt{chas}_i+\beta_5\,\mathtt{nox}_i^2+\beta_6\,\mathtt{rm}_i^2
    +\beta_7\,\mathtt{age}_i
    \label{eq:hedonic}
\end{equation}
\[
    \qquad+\beta_8\log(\mathtt{dis}_i)+\beta_9\log(\mathtt{rad}_i)+\beta_{10}\,\mathtt{tax}_i
    +\beta_{11}\,\mathtt{ptratio}_i+\beta_{12}\,\mathtt{b}_i+\beta_{13}\log(\mathtt{lstat}_i)+u_i,
\]
following the functional form of \citet{HarrisonRubinfeld1978}: crime rate (\texttt{crim}),
proportion of residential land zoned for large lots (\texttt{zn}), proportion of non-retail
business acreage (\texttt{indus}), a Charles River dummy (\texttt{chas}), squared nitrogen
oxide concentration (\texttt{nox}, an air-pollution measure), squared average number of
rooms (\texttt{rm}), proportion of pre-1940 housing stock (\texttt{age}), log weighted
distance to employment centers (\texttt{dis}), log highway accessibility index
(\texttt{rad}), property tax rate (\texttt{tax}), pupil-teacher ratio (\texttt{ptratio}),
a Black population share transform (\texttt{b}), and log share of the population with
lower socioeconomic status (\texttt{lstat}). OLS gives $R^2=0.806$, with all coefficients
of the expected sign: crime, pollution, distance to employment, and lower-status share
enter negatively, while rooms, river access, and highway accessibility enter positively.

\subsection{Omnibus Heteroskedasticity}

A Breusch--Pagan test \citep{BreuschPagan1979}, regressing $\hat u_i^2$ on the full set
of regressors in \eqref{eq:hedonic}, gives $\mathrm{LM}=95.3$ on $13$ degrees of freedom
($p<10^{-15}$): heteroskedasticity is overwhelmingly rejected at the model level. This is
the omnibus answer an applied researcher would obtain from standard diagnostics, and it is
silent on which coefficients are actually affected.

\subsection{Contrast-Specific Results}

\Cref{tab:boston} reports $\widehat R(a)$, the corresponding standard-error inflation
factor $\widehat S(a)$, and both classical and robust (HC0) standard errors for each
slope coefficient in \eqref{eq:hedonic}. The results span almost the entire range
documented in the Monte Carlo study: from $\widehat R(\mathtt{zn})=0.55$, where the robust
standard error is $26\%$ \emph{smaller} than the classical one, to $\widehat R(\mathtt{crim})=3.45$,
where the robust standard error is $86\%$ \emph{larger}.

\begin{table}[ht]
\centering
\caption{Contrast-specific heteroskedasticity relevance ratios in the Boston Housing hedonic regression
\eqref{eq:hedonic} ($n=506$). $\mathrm{SE(cl)}$ and $\mathrm{SE(rb)}$ are the classical
and HC0 standard errors. The $95\%$ confidence interval for $R(a)$ uses the feasible
$\widehat\Psi(a)$ from \Cref{thm:ratio_asynormal} with a quadratic auxiliary regression
of $\hat u_i^2$ on the contrast's own regressor (\Cref{rem:feasible_se}); $p$ is the
two-sided test of $H_0:R(a)=1$.}
\label{tab:boston}
\begin{tabular}{lrrrrrrr}
\toprule
Variable & $\hat\beta$ & SE(cl) & SE(rb) & $\widehat R(a)$ & $\widehat S(a)$ & 95\% CI for $R(a)$ & $p$ \\
\midrule
crim     & $-0.0119$ & $0.0012$ & $0.0023$ & $3.452$ & $1.858$ & $[0.378,\,6.526]$ & $0.118$ \\
zn       & $\phantom{-}0.0001$ & $0.0005$ & $0.0004$ & $0.554$ & $0.744$ & $[0.351,\,0.757]$ & $0.000$ \\
indus    & $\phantom{-}0.0002$ & $0.0023$ & $0.0018$ & $0.566$ & $0.752$ & --- & --- \\
chas     & $\phantom{-}0.0914$ & $0.0327$ & $0.0350$ & $1.142$ & $1.068$ & --- & --- \\
nox$^2$  & $-0.6381$ & $0.1116$ & $0.1199$ & $1.155$ & $1.075$ & --- & --- \\
rm$^2$   & $\phantom{-}0.0063$ & $0.0013$ & $0.0020$ & $2.306$ & $1.518$ & $[1.514,\,3.097]$ & $0.001$ \\
age      & $\phantom{-}0.0001$ & $0.0005$ & $0.0006$ & $1.243$ & $1.115$ & --- & --- \\
log(dis) & $-0.1913$ & $0.0329$ & $0.0389$ & $1.393$ & $1.180$ & --- & --- \\
log(rad) & $\phantom{-}0.0957$ & $0.0189$ & $0.0191$ & $1.029$ & $1.014$ & --- & --- \\
tax      & $-0.0004$ & $0.0001$ & $0.0001$ & $0.837$ & $0.915$ & --- & --- \\
ptratio  & $-0.0311$ & $0.0049$ & $0.0040$ & $0.654$ & $0.809$ & --- & --- \\
b        & $\phantom{-}0.0004$ & $0.0001$ & $0.0001$ & $2.016$ & $1.420$ & --- & --- \\
log(lstat)& $-0.3712$ & $0.0247$ & $0.0370$ & $2.254$ & $1.501$ & $[1.847,\,2.662]$ & $0.000$ \\
\bottomrule
\end{tabular}
\end{table}

The largest source of heteroskedasticity relevance is the crime-rate coefficient. Binning
observations by quartile of \texttt{crim}, the mean squared residual is essentially flat
across the bottom three quartiles ($0.018$, $0.016$, $0.019$) and then rises sharply, by
a factor of four, in the top quartile ($0.076$). Because directional leverage for the
\texttt{crim} coefficient is concentrated in exactly the tracts with extreme crime rates,
this concentrated variance translates directly into a large $\widehat R(\mathtt{crim})$
through the identity in \Cref{thm:directional_identity}: the handful of high-crime tracts
are simultaneously the most informative observations for $\hat\beta_{\mathtt{crim}}$ and
the noisiest. The $95\%$ confidence interval is wide, $[0.378,6.526]$, and the test of
$H_0:R(\mathtt{crim})=1$ does not reject at conventional levels ($p=0.118$): with only
$n=506$ observations and the relevant variation concentrated in a small subset of tracts,
the data pin down the \emph{existence} of an effect more clearly than its \emph{magnitude}.
This is itself a useful empirical lesson distinct from the omnibus test: the wide interval
signals that the practical consequence of heteroskedasticity for this coefficient, while
plausibly large, is not estimated with much precision, a distinction the Breusch--Pagan
statistic cannot make.

By contrast, the coefficient on \texttt{zn} has $\widehat R(\mathtt{zn})=0.554$ with a tight
confidence interval, $[0.351,0.757]$, that excludes one ($p<0.001$): for this coefficient,
robust inference is substantially less conservative than classical inference, while
the classical standard error appears overly conservative for this contrast, in the sense of \Cref{rem:validity}. The coefficients on
\texttt{rm}$^2$ and $\log(\texttt{lstat})$ — the two structural variables most central to
the hedonic specification — both show large, tightly estimated $\widehat R(a)$ values
above $2$, illustrating that heteroskedasticity is genuinely consequential for the two
coefficients an applied researcher would consider central to the regression's interpretation.

\subsection{Discussion}

This application illustrates the paper's central point in a setting where
heteroskedasticity is not a stylized feature of a simulated design but a
long-documented property of a widely used dataset. The omnibus
Breusch--Pagan test correctly establishes that heteroskedasticity is present
in the regression as a whole; it does not, and by construction cannot,
indicate that the practical consequence ranges from a $26\%$
\emph{reduction} in the standard error for \texttt{zn} to an $86\%$
\emph{increase} for \texttt{crim} within the same regression. An applied
researcher who runs a single heteroskedasticity test and then mechanically
reports HC0 standard errors throughout the table is, for roughly half the
contrasts in \Cref{tab:boston}, applying a correction that moves the
standard error by less than $20\%$ in either direction; for the other half,
the correction is substantively important and, in the case of \texttt{zn},
moves in the opposite direction from what default practice would assume.
Once heteroskedasticity is detected, the practically relevant question is
therefore not whether to use robust standard errors, but which contrasts are
materially affected by the covariance structure — and in which direction.
\section{Practical Guidance}
\label{sec:practical}

The framework developed in this paper translates directly into applied practice
through the sample heteroskedasticity relevance ratio $\widehat R(a)$, which can be computed from
any regression as the ratio of the sandwich variance estimate to the pooled variance
estimate for the contrast of interest. We suggest the following interpretation.

\paragraph{When $\widehat R(a)\approx1$.}
The heteroskedastic and pooled benchmarks are essentially equivalent for this contrast.
Robust and classical standard errors will be nearly identical, and \Cref{rem:validity}
shows there is no efficiency cost to using either. The researcher may use whichever
procedure is standard in the field without material consequence.

\paragraph{When $\widehat R(a)\gg1$.}
Heteroskedasticity substantially inflates the uncertainty for this contrast.
The robust standard error exceeds the classical one by a factor of $\widehat S(a)=\sqrt{\widehat R(a)}$.
Reporting the classical standard error would understate inference uncertainty.
Robust inference is materially preferable for this estimand, though \Cref{rem:validity}
notes that the increased sampling variability of the robust standard-error estimate itself
can reduce power relative to a (now invalid) classical benchmark, consistent with
\citet{KingRoberts2015}.

\paragraph{When $\widehat R(a)\ll1$.}
This is the less intuitive case. Heteroskedasticity is present but its pattern is
anti-aligned with the estimand: high-variance observations have low directional leverage
for $a'\beta$ and vice versa. The pooled or residual-bootstrap standard error is the
more conservative one. Applying robust inference actually produces \emph{narrower}
confidence intervals than the classical procedure. As \Cref{rem:validity} and \Cref{tab:size_efficiency} show concretely,
classical inference in this regime may become excessively conservative:
the standard error is inflated so much that the test can severely under-reject.
Thus, the researcher should not treat the larger classical standard error as
the safe default merely because it is larger.

\paragraph{Contrast specificity.}
Crucially, $\widehat R(a)$ is contrast-specific. The same regression may yield
$\widehat R(a_1)\gg1$ for one coefficient and $\widehat R(a_2)\approx1$ for another.
The appropriate diagnostic is to compute $\widehat R(a)$ for each estimand of interest,
not to apply a single omnibus heteroskedasticity test and then treat all coefficients uniformly.

\paragraph{Bootstrap choice.}
The ratio $\widehat R(a)$ also determines which bootstrap procedure should be preferred.
When $\widehat R(a)>1$, the pairs bootstrap is more conservative; when $\widehat R(a)<1$,
the residual bootstrap is more conservative. A researcher who systematically uses the pairs
bootstrap may be over-inflating standard errors for some contrasts while under-inflating
them for others within the same regression.

\paragraph{Formal inference on $\widehat R(a)$.}
\Cref{thm:ratio_asynormal} makes it possible to construct a confidence interval for
$R(a)$ and to test $H_0:R(a)=1$ (no effective heteroskedasticity for this estimand).
A rejection of this null signals that the choice between robust and classical inference
is material for the specific contrast under study, even if the test has no power for
contrasts orthogonal to the variance profile. Implementation requires a smooth estimate
of $\mathrm{Var}(u_i^2\mid X)$, for which an auxiliary regression of $\hat u_i^2$ on
functions of $x_i$ provides a convenient starting point (see \Cref{rem:feasible_se}).

\section{Conclusion}
\label{sec:conclusion}

This paper develops a contrast-specific theory of robust inference. The central finding
is that whether heteroskedasticity matters for a given estimand is determined by the
alignment between the error-variance profile and the directional leverage of the contrast.
The heteroskedasticity relevance ratio $R(a)=V_H(a)/V_R(a)$ summarizes this alignment in a single interpretable
number; it equals one under exact irrelevance, exceeds one when robust inference is more
conservative, and falls below one in the non-obvious case where the residual bootstrap or
classical inference is more conservative. When \(R(a)<1\), the larger classical standard error should not be interpreted
as providing a safer inferential procedure; it reflects a mismatch between the
location of uncertainty and the location of information for the estimand,
whereas the robust variance continues to target the correct asymptotic
variance. 
 Four main results were established: the exact
directional identity (\Cref{thm:directional_identity}), the information-weighting representation
(\Cref{prop:weighting}), a scalar asymptotic distribution via a Lindeberg--Feller argument and a delta-method step (\Cref{thm:ratio_asynormal}), and the generalization
to clustered and dependent settings (\Cref{prop:omega_decomp}). Applied to the Boston Housing
data of \citet{HarrisonRubinfeld1978} in \Cref{sec:empirical}, the framework shows that a
single, strongly significant omnibus heteroskedasticity test conceals a roughly seven-fold
range in $\widehat R(a)$ across coefficients in the same regression, from a $26\%$ reduction
in the standard error for one coefficient to an $86\%$ increase for another.

Extensions remain open. A nonparametric theory for the variance-function estimator
entering the feasible $\widehat\Psi(a)$ would complete the inferential framework. Extending
the directional identity to the off-diagonal term $\sum_{i\neq j}q_i(a)q_j(a)\Omega_{ij}$
under general spatial, network, or factor dependence is a natural next step. Finite-sample
refinements via Edgeworth expansions may sharpen coverage below $n=400$.

\appendix
\section{Proofs of Auxiliary Results}
\label{app:proofs}

\subsection{Proof of \Cref{prop:cov_relevance}}
\label{app:thm1}

Premultiplying \eqref{eq:AL} by $a'$ gives
\[
    \sqrt n(a'\hat\theta-a'\theta_0)
    =
    \frac{1}{\sqrt n}\sum_{i=1}^n a'\psi_i+o_p(1).
\]
The asymptotic variance is therefore $a'\Omega a$.
Substituting $\Omega=\Omega_0+\Delta$ yields $a'\Omega a=a'\Omega_0a+a'\Delta a$.
The difference between the perturbed and baseline asymptotic variances is $a'\Delta a$, proving the result.
\hfill$\square$

\subsection{Proof of \Cref{lem:sum_ell}}
\label{app:lem1}

By definition, $\sum_{i=1}^n\ell_i(a)=a'(X'X)^{-1}(\sum_ix_ix_i')(X'X)^{-1}a
=a'(X'X)^{-1}(X'X)(X'X)^{-1}a=a'(X'X)^{-1}a$.
\hfill$\square$

\subsection{Proof of \Cref{prop:weighting}}
\label{app:prop1}

From \Cref{thm:directional_identity}, $V_H(a)=\sum_i\ell_i(a)\sigma_i^2$ and $V_R(a)=\bar\sigma^2\sum_i\ell_i(a)$. Dividing and substituting $w_i(a)=\ell_i(a)/\sum_j\ell_j(a)$ gives \eqref{eq:weighting}. Subtracting one from both sides gives \eqref{eq:R_minus_one}.
\hfill$\square$

\subsection{Proof of \Cref{cor:bounds}}
\label{app:cor_bounds}

By \Cref{prop:weighting}, $R(a)=\sum_iw_i(a)\cdot(\sigma_i^2/\bar\sigma^2)$, a convex combination of the ratios $\sigma_i^2/\bar\sigma^2$ with weights $w_i(a)\ge0$ summing to one. The inequality follows directly.
\hfill$\square$

\subsection{Proof of \Cref{prop:omega_decomp}}
\label{app:prop2}

Using \eqref{eq:Vomega},
\[
    V_\Omega(a)=
    \sum_{i,j}
    \{a'(X'X)^{-1}x_i\}\{a'(X'X)^{-1}x_j\}\Omega_{ij}
    =
    \sum_{i,j}q_i(a)q_j(a)\Omega_{ij}.
\]
Separating diagonal ($i=j$) and off-diagonal ($i\neq j$) terms gives \eqref{eq:omega_decomp}.
Since $V_R(a)=\bar\sigma^2\sum_i\ell_i(a)=\bar\sigma^2\sum_iq_i(a)^2$, subtracting yields \eqref{eq:omega_minus_pooled}.
\hfill$\square$

\subsection{Matrix form of \Cref{thm:directional_identity}}
\label{app:matrix}

Setting $V_H=(X'X)^{-1}X'\Sigma X(X'X)^{-1}$ and $V_R=\bar\sigma^2(X'X)^{-1}$,
\[
    V_H-V_R
    =
    (X'X)^{-1}X'(\Sigma-\bar\sigma^2 I_n)X(X'X)^{-1}.
\]
Premultiplying and postmultiplying by $a$ gives $V_H(a)-V_R(a)=\sum_i\ell_i(a)(\sigma_i^2-\bar\sigma^2)$, which is \eqref{eq:directional_identity}.

\subsection{Directional leverage and White-type regressors}
\label{app:white}

Let $c=(X'X)^{-1}a$. Then $\ell_i(a)=(c'x_i)^2$. The directional covariance in \Cref{thm:directional_identity} is therefore the empirical covariance between $\hat u_i^2$ and the square of the scalar index $c'x_i$, in contrast to an omnibus heteroskedasticity diagnostic that tests whether $\hat u_i^2$ depends on all components and cross-products of $x_i$.

\subsection{Proof of \Cref{prop:ratio_consistency}}
\label{app:prop_consistency}

Let $\widetilde V_H(a)=\sum_{i=1}^n\ell_i(a)u_i^2$ be the oracle version of $\widehat V_H(a)$
(using true errors rather than OLS residuals).

\paragraph{Step 1: Oracle variance.}
Conditional on $X$, $\E\{\widetilde V_H(a)\mid X\}=V_H(a)$ and
\[
    \Var\{\widetilde V_H(a)\mid X\}
    =
    \sum_{i=1}^n\ell_i(a)^2\Var(u_i^2\mid X)
    \le C\sum_{i=1}^n\ell_i(a)^2.
\]
The key inequality
\[
    \sum_i\ell_i(a)^2\le \max_i\ell_i(a)\cdot\sum_i\ell_i(a)
    =
    \max_i w_i(a)\cdot\bigl(a'(X'X)^{-1}a\bigr)^2
\]
together with $a'(X'X)^{-1}a=O(n^{-1})$ and $\max_iw_i(a)\to0$ (\Cref{ass:leverage}) gives
$\Var\{\widetilde V_H(a)\mid X\}=o(n^{-2})$.
By Chebyshev's inequality,
\[
    \widetilde V_H(a)-V_H(a)=o_p\!\left(a'(X'X)^{-1}a\right)=o_p(n^{-1}).
\]

\paragraph{Step 2: Residual estimation error in $\widehat V_H(a)$.}
Writing $\hat u_i-u_i=-x_i'(\hat\beta-\beta)$,
\[
    \hat u_i^2-u_i^2
    =
    -2u_ix_i'(\hat\beta-\beta)+\{x_i'(\hat\beta-\beta)\}^2.
\]
By \Cref{ass:design,ass:errors}, $\|\hat\beta-\beta\|=O_p(n^{-1/2})$, since
$\hat\beta-\beta=(X'X)^{-1}X'u$ and $\lambda_{\min}(n^{-1}X'X)\ge c>0$.

\emph{Linear term.} Two applications of the Cauchy--Schwarz inequality give
\[
    \left|\sum_i\ell_i(a)u_ix_i'(\hat\beta-\beta)\right|
    \le
    \|\hat\beta-\beta\|\!
    \left(\sum_i\ell_i(a)^2\|x_i\|^2\right)^{1/2}\!\!
    \left(\sum_i u_i^2\right)^{1/2}.
\]
Using $\sum_i\ell_i(a)^2\|x_i\|^2\le\sup_i\|x_i\|^2\max_iw_i(a)\cdot(a'(X'X)^{-1}a)^2$,
$a'(X'X)^{-1}a=O(n^{-1})$, and $\sum_iu_i^2=O_p(n)$ (from $\bar\sigma^2=O(1)$ and
independence), the linear term equals
\[
    O_p(n^{-1/2})\cdot O(n^{-1})\sqrt{\sup_i\|x_i\|^2\max_iw_i(a)}\cdot O_p(n^{1/2})
    =O_p(n^{-1})\sqrt{\sup_i\|x_i\|^2\max_iw_i(a)}
    =o_p(n^{-1}),
\]
by \Cref{ass:regressor_magnitude}.

\emph{Quadratic term.} Directly,
\[
    \sum_i\ell_i(a)\{x_i'(\hat\beta-\beta)\}^2
    \le\|\hat\beta-\beta\|^2\sup_i\|x_i\|^2\sum_i\ell_i(a)
    =O_p(n^{-1})\cdot\sup_i\|x_i\|^2\cdot O(n^{-1}).
\]
Under \Cref{ass:regressor_magnitude} and \Cref{ass:leverage},
$\sup_i\|x_i\|^2\cdot\max_iw_i(a)\to0$ and $\max_iw_i(a)\ge n^{-1}$ (since weights sum to one),
so $\sup_i\|x_i\|^2=o(n)$.  Hence the quadratic term is $O_p(n^{-1})\cdot o(n)\cdot O(n^{-1})=o_p(n^{-1})$.

Combining, $\widehat V_H(a)-\widetilde V_H(a)=o_p(n^{-1})$, so $\widehat V_H(a)-V_H(a)=o_p(n^{-1})$.

\paragraph{Step 3: Pooled variance.}
Write $\hat u_i^2-u_i^2=-2u_ix_i'(\hat\beta-\beta)+\{x_i'(\hat\beta-\beta)\}^2$.
Taking the average, $n^{-1}\sum_i(\hat u_i^2-u_i^2)=T_1+T_2$ where
\[
    |T_1|=\frac{2}{n}\left|\sum_i u_ix_i'\right|\cdot\|\hat\beta-\beta\|
    =O_p(n^{-1/2})\cdot O_p(n^{-1/2})=O_p(n^{-1}),
\]
using $n^{-1/2}\|\sum_ix_iu_i\|=O_p(1)$ and $\|\hat\beta-\beta\|=O_p(n^{-1/2})$, and
\[
    |T_2|
    \le
    \|\widehat\beta-\beta\|^2\sup_i\|x_i\|^2
    =
    O_p(n^{-1})\sup_i\|x_i\|^2
    =
    o_p(1),
\]
Hence $n^{-1}\sum_i\hat u_i^2=n^{-1}\sum_iu_i^2+o_p(1)$.
Since $n^{-1}\sum_iu_i^2\p\bar\sigma^2$ by the law of large numbers
(the $u_i$ are uncorrelated with uniformly bounded variance, so $\bar\sigma^2=O(1)$),
we have $\hat\sigma^2=n^{-1}\sum_i\hat u_i^2\p\bar\sigma^2$, so $\hat\sigma^2-\bar\sigma^2=o_p(1)$, and
\[
    \widehat V_R(a)-V_R(a)=(\hat\sigma^2-\bar\sigma^2)\,a'(X'X)^{-1}a
    =o_p(1)\cdot O(n^{-1})=o_p(n^{-1}).
\]

\paragraph{Step 4: Ratio convergence.}
Since $V_R(a)=\bar\sigma^2 a'(X'X)^{-1}a\ge c\cdot cn^{-1}\cdot c=c^3n^{-1}>0$ (from
\Cref{ass:design,ass:errors}), the map $(u,v)\mapsto u/v$ is continuously differentiable in
a neighborhood of $(V_H(a),V_R(a))$. With $\widehat V_H(a)-V_H(a)=o_p(n^{-1})$ and
$\widehat V_R(a)-V_R(a)=o_p(n^{-1})$, the continuous-mapping theorem gives
$\widehat R(a)-R(a)\p0$. The result for $\widehat S(a)=\sqrt{\widehat R(a)}$ follows again
by the continuous-mapping theorem.
\hfill$\square$

\subsection{A sufficient primitive condition for \Cref{ass:residual_remainder_clt}}
\label{app:clt_remainder}

This subsection explains why \Cref{ass:residual_remainder_clt} is a standard high-level condition. Suppose, in addition to \Cref{ass:design,ass:errors}, that $\sup_i\|x_i\|\le C$ and $\max_i\ell_i(a)=O(n^{-2})$. Then $\widehat V_H(a)-\widetilde V_H(a)=o_p(n^{-3/2})$ and the analogous pooled-variance remainder in \Cref{ass:residual_remainder_clt} is $o_p(n^{-3/2})$.

Indeed,
\[
\widehat V_H(a)-\widetilde V_H(a)
=
-2(\widehat\beta-\beta)'\sum_{i=1}^n\ell_i(a)x_i u_i
+(\widehat\beta-\beta)'\left(\sum_{i=1}^n\ell_i(a)x_ix_i'\right)(\widehat\beta-\beta).
\]
By \Cref{ass:design,ass:errors}, $\|\widehat\beta-\beta\|=O_p(n^{-1/2})$. Conditional on $X$,
\[
\Var\left(\sum_i\ell_i(a)x_i u_i\mid X\right)
\le C\sum_i\ell_i(a)^2\|x_i\|^2
=O(n^{-3}),
\]
so $\sum_i\ell_i(a)x_i u_i=O_p(n^{-3/2})$. The linear part is therefore $O_p(n^{-2})$. For the quadratic part,
\[
\left\|\sum_i\ell_i(a)x_ix_i'\right\|
\le C\sum_i\ell_i(a)=O(n^{-1}),
\]
so the quadratic part is $O_p(n^{-2})$. Hence $\widehat V_H(a)-\widetilde V_H(a)=O_p(n^{-2})=o_p(n^{-3/2})$.

For the pooled variance,
\[
\widehat V_R(a)-V_R(a)-A_n(a)n^{-1}\sum_i(u_i^2-\sigma_i^2)
=
A_n(a)n^{-1}\sum_i(\hat u_i^2-u_i^2).
\]
The average residual-estimation error satisfies $n^{-1}\sum_i(\hat u_i^2-u_i^2)=O_p(n^{-1})$ under the same fixed-design conditions, while $A_n(a)=O(n^{-1})$. Thus this remainder is $O_p(n^{-2})=o_p(n^{-3/2})$.

\subsection{Proof of \Cref{thm:ratio_asynormal}}
\label{app:thm_asynormal}

Let $A_n(a)=a'(X'X)^{-1}a$, $z_i=u_i^2-\sigma_i^2$, and $b_i(a)=\ell_i(a)-R(a)A_n(a)/n$.

\paragraph{Step 1: Exact decomposition.}
The identity $V_H(a)=R(a)V_R(a)$ gives the algebraically exact representation
\begin{equation}
    \widehat R(a)-R(a)
    =
    \frac{\widehat V_H(a)-R(a)\widehat V_R(a)}{\widehat V_R(a)}.
\label{eq:exact_id}
\end{equation}
By \Cref{ass:residual_remainder_clt},
\[
    \widehat V_H(a)-V_H(a)=\sum_{i=1}^n\ell_i(a)z_i+o_p(n^{-3/2}),
    \qquad
    \widehat V_R(a)-V_R(a)=\frac{A_n(a)}{n}\sum_{i=1}^n z_i+o_p(n^{-3/2}).
\]
Substituting into the numerator of \eqref{eq:exact_id} and using $V_H(a)=R(a)V_R(a)$:
\begin{align*}
    \widehat V_H(a)-R(a)\widehat V_R(a)
    &=\bigl[\widehat V_H(a)-V_H(a)\bigr]-R(a)\bigl[\widehat V_R(a)-V_R(a)\bigr]\\
    &=\sum_{i=1}^n\Bigl[\ell_i(a)-R(a)\frac{A_n(a)}{n}\Bigr]z_i+o_p(n^{-3/2})
    =\sum_{i=1}^n b_i(a)z_i+o_p(n^{-3/2}).
\end{align*}
Inserting back into \eqref{eq:exact_id}:
\begin{equation}
    \widehat R(a)-R(a)
    =
    \frac{\sum_{i=1}^n b_i(a)z_i+o_p(n^{-3/2})}{\widehat V_R(a)}.
\label{eq:R_decomp}
\end{equation}

\paragraph{Step 2: Scalar CLT.}
The summands $b_i(a)z_i$ are conditionally independent and mean-zero given $X$
(since $\E(z_i\mid X)=0$), with conditional variance
$b_i(a)^2\Var(u_i^2\mid X)$ summing to $s_n^2(a)$.
By \Cref{ass:scalar_lindeberg}, $\Psi_n(a)\to\Psi(a)\in(0,\infty)$, so $s_n(a)>0$ for all large $n$.
The Lindeberg condition in \Cref{ass:scalar_lindeberg} is exactly that for the
scalar triangular array $\{b_i(a)z_i/s_n(a)\}_{i=1}^n$.
The Lindeberg--Feller central limit theorem \citep[Theorem~27.2]{Billingsley1995} gives
\begin{equation}
    \frac{\sum_{i=1}^n b_i(a)z_i}{s_n(a)}\dlim N(0,1).
\label{eq:clt_scalar}
\end{equation}

\paragraph{Step 3: Combining and Slutsky.}
Since $\sqrt{\Psi_n(a)}=\sqrt{n}\,s_n(a)/V_R(a)$ and $V_R(a)>0$, dividing
\eqref{eq:R_decomp} by $\sqrt{\Psi_n(a)/n}$ gives
\[
    \frac{\sqrt{n}\{\widehat R(a)-R(a)\}}{\sqrt{\Psi_n(a)}}
    =
    \frac{V_R(a)}{\widehat V_R(a)}
    \frac{\sum_{i=1}^n b_i(a)z_i}{s_n(a)}
    +
    \frac{\sqrt n\, r_n}{\widehat V_R(a)\sqrt{\Psi_n(a)}},
\]
where \(r_n=o_p(n^{-3/2})\). Since
\[
    \widehat V_R(a)=O_p(n^{-1}),
    \qquad
    \Psi_n(a)\to\Psi(a)\in(0,\infty),
\]
we have
\[
    \frac{\sqrt n\, r_n}{\widehat V_R(a)\sqrt{\Psi_n(a)}}
    =
    \frac{o_p(n^{-1})}{O_p(n^{-1})}
    =
    o_p(1).
\]
By \Cref{prop:ratio_consistency}, $V_R(a)/\widehat V_R(a)\p1$.
Combined with \eqref{eq:clt_scalar} and Slutsky's theorem,
\[
    \frac{\sqrt{n}\{\widehat R(a)-R(a)\}}{\sqrt{\Psi_n(a)}}
    \dlim
    N(0,1).
\]
\paragraph{Step 4: Result for $\widehat S(a)$.}
By the delta method applied to $h(r)=\sqrt{r}$ at $r=R(a)>0$ with $h'(R(a))=1/(2\sqrt{R(a)})$,
\[
    \sqrt{n}\{\widehat S(a)-S(a)\}
    =
    \frac{1}{2\sqrt{R(a)}}\sqrt{n}\{\widehat R(a)-R(a)\}+o_p(1),
\]
which gives $\sqrt{n}\{\widehat S(a)-S(a)\}/\sqrt{\Psi_n(a)/(4R(a))}\dlim N(0,1)$.
\hfill$\square$

\bibliographystyle{plainnat}
\bibliography{references}

\end{document}